\documentclass[11pt]{article}
\usepackage[utf8]{inputenc}

\usepackage{graphicx,amsthm,amssymb,amsmath,authblk}
\usepackage{hyperref}
\usepackage{color}
\usepackage{algorithm}
\usepackage{algpseudocode}
\usepackage{verbatim}
\usepackage{subfig}
\usepackage{listings}
\usepackage{booktabs}

\newtheorem{definition}{Definition}
\newtheorem{theorem}{Theorem}
\newenvironment{keywords}{%
    \par\medskip\noindent
    \small
    \textbf{Keywords:}\enspace\ignorespaces
}{%
    \par\medskip
}

\title{{\bf Generating gaussian pseudorandom noise with binary sequences}}
\author[1]{Francisco-Javier Soto \thanks{Email: franciscojavier.soto@urjc.es}}
\author[2]{Ana I. Gómez\thanks{Email: ana.gomez.perez@urjc.es}}
\author[1]{Domingo Gómez-Pérez\thanks{Email: gomezd@unican.es}}
\affil[1]{Universidad Rey Juan Carlos, Madrid, Spain}
\affil[2]{Universidad de Cantabria, Santander, Spain}
\date{}

\begin{document}
\maketitle
\begin{abstract}
 
 Gaussian random number generators attract a widespread interest due to their applications in several fields.
  Important requirements include easy implementation, tail accuracy,  and, finally, a flat spectrum. In this work, we study the applicability of uniform pseudorandom binary generators in combination with the Central Limit Theorem to propose an easy to implement, efficient and flexible algorithm that leverages the properties of the pseudorandom binary generator used as an input, specially with respect to the correlation measure of higher order, to guarantee the quality of the generated samples. Our main result provides a relationship between the pseudorandomness of the input and the statistical moments of the output.
We propose a design  based on the combination of pseudonoise sequences commonly used on wireless communications with known hardware implementation, which can generate sequences with guaranteed  statistical distribution properties sufficient for many real life applications and simple machinery.
Initial computer simulations on this construction show promising results in the quality of the output and the computational resources in terms of required memory and complexity. 
\end{abstract}
\begin{keywords}
Pseudorandom number generator, Gold code, $m$-sequence, Central Limit Theorem, Gaussian distribution.
\end{keywords}
\section{Introduction}
\label{sec:intro}
The performance of many industrial applications depends on the simulation of random events, for example biological simulations, communication channels measurement or electronic instrument calibration. Although one of the most commonly demanded statistical distribution is the Gaussian distribution, the principal body of work is focused in the uniform distribution case. There are well-known results about construction of pseudorandom numbers with uniform distribution \cite{kuipers2012uniform}, that have been proposed as the basis to generate any other distribution. In a general setting, some methods rely on the rejection method and variations of it, such as acceptance-rejection, composition rejection, rejection with squeeze, etc. Others are based on the inversion method, or use special properties of the normal distribution like the Box-Muller algorithm that is widely used in practice due to its improved performance~\cite{hormann2004automatic}.
However, there are downsides which limit the applicability of this algorithm: it requires computation of elementary functions including the logarithm and square root as well as trigonometric functions which are costly to implement with logic gates. Moreover, the tail accuracy is directly dependent on the implementation \cite{malik2016gaussian}.

The previous limitations show just a few  of the found technical difficulties to implement a Gaussian Random Number generator (GRNG) over hardware. 
This is of maximum importance for many different real world applications where memory, computation time and throughput are constrained.  In this case hardware-related parameters (v.g., number of logical gates, buffers, circuit layout,...) have to be minimized and it is recommended to reuse circuitry and functions, when possible. Different hardware-based techniques
have been compared based on the amount of required  hardware resources, the statistical precision and the tail accuracy \cite{alimohammad2008compact}.


A straightforward implementation of sums of truly random samples tends to the Gaussian population due to the well known Central Limit Theorem (CLT), although the resulting distribution will approximate poorly in the tail to the Gaussian distribution. The ease of implementation of this approach has fostered the proposal of CLT-based GRNGs with improved statistical properties and tail accuracy such as CLT correction, CLT inversion and multi hat methods \cite{malik2016gaussian}.  Still, there is room for improvement on the performance and speed by limiting the statistical accuracy properties, aiming at achieving a minimalist solution that maintains an acceptable flexibility and ease of implementation.

In this work, we study the theoretical framework to apply the CLT to a sum of pseudorandom binary sequences with good correlation properties. We note that this approach provides a worse approximation comparing  with of uniform numbers for the same number of terms (see Irwin-Hall distribution of the sums of numbers in the interval $(0, 1)$ \cite{marengo2017geometric}). Still, it is widely seen as an acceptable trade off because the former allows to work at bit level. We note that also there are already efficient implementation for most pseudorandom sequences because of the widespread use in wireless communications.

This work focus on ``simple'' GRNG, which requires to control moments up to third and fourth order, i.e., skewness and kurtosis,  that have been studied for $m$-sequences~\cite{wainberg1970subsequences}. The family of $m$-sequences is a good candidate due to the easy implementation by Linear Feedback Shift Registers (LFSRs). However $m$-sequences and several derived families show peaks on the higher order correlations~\cite {chen-2022-correl-measur, waifi2023}, that will affect the quality of the output of any CLT based GRNG as we will proof in Section ~\ref{sec:GeneratingG}.
Increasing the number of LFSRs and at the same time, using their states as part of the input can reduce previous limitations. Recent works \cite{kang2010fpga, condo2015pseudo, cotrina2020gaussian, cotrina2021gaussian} improve previous GRNGs architectures by using this idea and validate the results through computer experiments. This heuristic approach seems to also have nice properties, such that reduced latency and allowing parallelism. This study is the first step to fill this gap between theory and practice.

The outline of this work is the following: Section \ref{sec:GeneratingG} describes the main results regarding  the moments  of Gaussian pseudorandom sequences constructed from the combination of binary sequences, 
 Section \ref{sec: gold_codes} shows that Gold codes under certain conditions guarantee the absence of full third and fourth peaks. Section~\ref{sec: computations}
provides computational experiments for $m$-sequences and Gold codes and
Section~\ref{sec: conclusions} concludes the article with a
discussion on the parameters, a comparison with other LFSR-based GRNGs and some open problems.

\section{ Gaussian Random Number Generation from Binary Sequences}
\label{sec:GeneratingG}




This section shows the dependence between the moments of Gaussian pseudorandom sequences and the correlation measure of the binary sequences used to generate them. In the following Theorem \ref{th: bounds_moments} we give this dependence in terms of bounds. We denote the binary sequences of period $N$ as $s(i) \in \{-1, 1\}$ or simply $s$ when possible. Also, for the reader's convenience, we recall the definition of \emph{combined correlation measure of
order~$k$} for the periodic case~\cite{gyarmati2004family}.
\begin{definition}
\label{def:combined_correlation}
Given a binary sequence~$s$ of period~$N$,
$\theta_k(s, N)$ is the combined correlation measure of order~$k$, defined as
\begin{equation}
\label{eq: corr}
\theta_k(s, N) = \max_{L, D, T}\left|\sum_{i=1}^T
s(L\cdot i+d_1) \cdots s(L\cdot i+d_{k}) \right |,
\end{equation}
where $D=(d_1,\ldots, d_k)$ with
$0\le d_1 < \cdots < d_k < N$, the sum on $i$ run such  all values $L\cdot i+d_1,\ldots, L\cdot i+d_k\in\{1,\ldots, N\}$ and $T \le N$.
\end{definition}
The combined correlation measure of order~$k$ is a powerful measure for asserting the pseudorandomness of a binary sequence, which calculates the correlation over arithmetic subsequences. 

This is the discrete version of  the product moments that have been characterized for continuous Gaussian random variables  \cite{song2015explicit, song2017proof}. 
For a continuous Gaussian random variable with zero mean, the product moments of odd order $k$ must be exactly zero \cite[Corollary 2]{song2015explicit}.  This also holds for every order $k$ when the random variables are different~\cite[Remark 5]{song2015explicit}.  In particular, we have proven that if the correlation measure of the binary sequence $s$ is well-bounded and $M$ is much smaller than   $T$ the generated sequences satisfy those facts.


\begin{theorem}
\label{th: bounds_moments}
    Let $s(i)$ be a binary sequences of period $N$, $M$ a  positive integer with $M \ll N$, and $k$ a non-negative integer. Then the following holds,
\begin{multline}
\label{eq: rows}
\frac{1}{T}\sum_{i = 1}^{T} \left( \sum_{n = 1}^{M}s(i + n) \right)^{k}
              \leq (M(k-1))^{k/2} + 
              \frac{M^{k}\max_{1\leq r \leq k}\theta_{r} (s(i), N)}{T},
\end{multline}
where if $k$ is odd, the first term in the right of the inequality disappears.

\begin{proof}
We follow the argument of Davenport and Erd\"os \cite[Lemma 3]{davenport1952distribution}. First, we prove Equation \eqref{eq: rows}. Expanding the 
\begin{multline}
    \label{large_sum}
\sum_{i = 1}^{T} \left( \sum_{n = 1}^{M}s(i + n) \right)^{k} = \\
\sum_{d_1 = 1}^{M} \cdot \cdot \cdot \sum_{d_{k} = 1}^{ M} \sum_{i = 1}^{T} s(i + d_1) \cdot \cdot \cdot s(i + d_{k})\le \\
\sum_{d_1 = 1}^{M} \cdot \cdot \cdot \sum_{d_{k} = 1}^{ M} \left |\sum_{i = 1}^{T} s(i + d_1) \cdots s(i + d_{k})\right |
\end{multline}
we can bound the inner sum depending on the integers $d_1,\ldots,d_{k}$. If each value of $d_1,\ldots, d_k$ is repeated an even number of times, we bound the inner sum by $T$. The number of choices for $d_1,\ldots, d_k$ so that to happen is less than $(M (k-1))^{k/2}.$

When not all $d_1,\ldots, d_k$ appear an even number of times, 
we can remove repetitions, giving a subset  different $d'_1,...,d'_{r}$ which substituting in the inner sum 
\begin{equation}
\label{cote_correlation}
\left |\sum_{i = 1}^{T} s(i + d'_1) \cdot \cdot \cdot s(i + d'_{r}) \right | \leq \theta_{r}(s(i), N).
\end{equation}
The number of such cases is bounded by $M^k$, therefore the contribution in Equation \eqref{large_sum} is less than $$M^{k}\max_{1\leq r\leq k}\theta_{r}(s(i), N).$$
This finishes the proof.
\end{proof}
\end{theorem}


We define the following sequence $S$ which values distributes following a Gaussian distribution by the CLT:

\begin{equation}
\label{eq: gaussian_source}
 S(i) = \left( M^{-1/2}\left(\sum_{n = 1}^{M}s\left(n + i M\right) \right)\right).
\end{equation}

This definition aims to minimize the dependence between consecutive terms  of the generated Gaussian sequence by 
taking sums of blocks of $M$ terms without reusing them. We emphasize that for every $1 \leq M \leq N$ similar properties that for Equation \eqref{eq: rows} must hold,  where
there are overlaps between each consecutive sum of $M$ terms.




\begin{theorem}
\label{th: th2}
Let $S$ be the sequence defined in the Equation \eqref{eq: gaussian_source}, $M$ a non-negative integer with $M \ll T$. Let us take $k$ positive integers $0\le d_1<\ldots< d_k< T-M$. Under these conditions we have the following for every $M$ such that $0 \leq M \leq T$,
\begin{multline}
\left |\frac{1}{T}\sum_{i =1}^T S(i + d_1) \cdot \cdot \cdot S(i + d_k) \right | \leq 
(k-1)^{k/2} +
 \frac{M^{k/2}\max_{1\leq r \leq k}\theta_{r} (s(i), N)}{T},
\end{multline}
where the first term in the right part of the inequality dissapears for $k$ odd.  
\end{theorem}
\begin{proof}
The proof is similar to Theorem \ref{th: bounds_moments}. First,
\begin{multline*}
\left |\sum_{i =1}^T S(i + d_1) \cdot \cdot \cdot
S(i + d_k) \right |= \\
\frac{1}{M^{k/2}}\left|\sum_{n_1 = 1}^M\cdot \cdot \cdot \sum_{n_k = 1}^M \sum_{i=1}^T \left(\prod_{j=1}^ks(Mi + Md_j+ n_j) \right) \right|\le \\
\frac{1}{M^{k/2}}\sum_{n_1 = 1}^M\cdot \cdot \cdot \sum_{n_k = 1}^M \left|\sum_{i=1}^T \left(\prod_{j=1}^ks( Mi+ M d_j+ n_j) \right) \right|.
\end{multline*}
Now, we expand the inner sum,
\begin{multline*}
    \left|\sum_{i=1}^T \left(\prod_{j=1}^ks( Md_j + iM+ n_j) \right) \right|
= \\ \left |\sum_{i=1}^T s(Mi + M d_1 + n_1)\cdots s(Mi + M d_k + n_k) \right |.
\end{multline*}

 The number of possible ways to take the integers $n_1,\ldots,n_k$ in the set $\{1,...,M\}$ is $M^k$. Therefore if $k$ is odd we can bound this internal sum by $M^k\max_{1\leq r \leq k}\theta_{r} (s(i), N).$

 If $k$ is even, it can happen that $M d_1+ n_1,\ldots, M d_k+ n_k$ appears an even number of times each value, therefore making the sum equal to $T$.
 We can calculate the number of times that this happens as in Theorem~\ref{th: bounds_moments}. This finishes the proof.
\end{proof}

\subsection{Correlation properties of Gold codes}
\label{sec: gold_codes}
There are several families of binary sequences with good bounds for the combine correlation measure for many values of $k$, we will focus on sequences generated by two LFSRs such as the Gold code as a proof of concept due to its simple implementation with only two LFSRs  and  an XOR gate \cite{sarwate1980crosscorrelation}.

 For a finite field of characteristic 2, denoted by $\mathbb{F}_q$ with $q = 2^n$, the \emph{trace function} is defined as the following map,
$\text{Tr}: \mathbb{F}_q \rightarrow \mathbb{F}_2$,
$$
\text{Tr}(x) = \sum_{j = 1}^n x^{2^j}.
$$
\begin{definition}
\label{def: gold}
    Let $\alpha \in \mathbb{F}_q$ be a primitive element and let $f(x) = x + x^{2^r + 1}$ where $(r , n) = 1$ be a polynomial in the $\mathbb{F}_q$. A \emph{Gold code} is the following binary sequence $s$ with period $q-1$,
\begin{equation}
\label{eq:trace_general_sequence}
s(i) = \psi(f(\alpha^i)),
\end{equation}
where $\psi(x) = (-1)^{\text{Tr}(x)}$, i.e., it is an additive character defined by the trace function.
\end{definition}

 We remark that to use the LFSRs architecture, we must convert the binary {0,1} sequence to {-1,1} as in Equation \eqref{eq:trace_general_sequence}.

Now we compile previous results from \cite {chen-2022-correl-measur, waifi2023,follath2008construction} in Theorem~\ref{th: theorem_Gold}.
These results on the correlation measure of Gold codes holds for any $n$.
\begin{theorem}
\label{th: theorem_Gold}
Let $\mathbb{F}_q$ be a finite field where $q - 1 = 2^n -1$ is a Mersenne prime and $s(i)$ be a Gold code. Then, for every $k$ such that $1\leq k \leq 4$
$$
\theta_k(s, q-1) \leq 9 n 2^{2r+1 + n/2}.
$$
Also, there is a full peak at $k = 5$. 
\end{theorem}


\section{Computational Experiments}
\label{sec: computations}
In this section we compare different GRNGs using the CLT and the Tausworthe model.

Let $f(x) = \sum _{i = 0}^n b_i x ^i$ be a polynomial with coefficients in $\mathbb{F} _2$ such that $b_0 = b_n = 1$ be a primitive polynomial of the finite field $\mathbb {F}_q$  \textit{i.e.}, the minimal polynomial of a primitive element $\alpha$ of $\mathbb{F}_q$. We recall the concept of maximum length LFSR for the reader's convenience. 

\begin{definition}
\label{def: lfsr}
Using the preceding notation, a maximum length LFSR of $n$ registers, consists of an initial state $\vec{e_0} \in \mathbb{F}_2^n$, a transition function, $T: \mathbb{F}_2^n \rightarrow \mathbb{F}_2^n$, such that
$$
T(e_1,...,e_n) = \left(e_2, ..., e_n, \sum_{i = 1} ^n e_i b_{i - 1} \right)
$$
where $b_i's$ are the coefficients of $f(x)$, which is called the characteristic polynomial of the LFSR, and an output function defined by
$$
\text{out}\left(T^j(\vec{e_0})\right) = \text{out}(e_{1+ j}, ..., e_{n+ j}) =  e_{1+ j} \hspace{0.1in} \text{for every} \hspace{0.1in}  j \geq 0,
$$
with 
$$
T^0(\vec{e_0}) = \vec{e_0} = (e_1,...,e_n).
$$
\end{definition}

The Tausworthe model is a frequently used model to construct pseudorandom numbers with uniform distribution. Each state $\vec{e_i}$ defines a real number in the interval $(0, 1)$ in binary notation by the following application:
\begin{eqnarray*}
\varphi: &\mathbb{F}_2^B\mapsto &[0,1)\\
 & (e_1,\ldots, e_B)\mapsto& \sum_{i=1}^B e_i 2^{-i}. 
\end{eqnarray*}

 Different states can be also combined or partially taken by the application depending on the bit depth $B$, which is the number of registers used to define the uniform pseudorandom numbers.




 We consider the following binary sequences for the experiments proposed from practice. First, we consider an $m$-sequence whose characteristic polynomial is $x^{89} + x^{38} + 1$. Second, we consider a Gold code with $r = 1$ in the Definition \ref{def: gold} whose characteristic polynomials are $f_1(x) = x^{89} +x^{38} +1$ and $f_2(x) =  x^{89} +x^{72} + x^{55} + x^{38} + 1$.

 The parameters for Equation \eqref{eq: gaussian_source} are  $M = 256$ and the Tausworthe model has $32$ bit depth, i.e., taking eight sums of consecutive 32-bit numbers. We have measure the first four moments with a sample size of $10^5$ for both models. The results are summarized in  Table \ref{tab:moments} for $T=10^5$.
 



These tests have been consider in wireless communications, where  signals are interfered  by  delayed version by different shifts. This is studied through the product moments and the polyspectrum \cite{tugnait1994detection, green2003utility}. The triple product moments normalized by the sample size $T$ of each  considered configuration is shown in Figure \ref{fig: triple_correlations} in a window $100 \times 100$ .



\begin{table}[htbp]
\centering
\caption{Moments for the both models and using different binary sequences with $M = 256$.}
\label{tab:moments}
\begin{tabular}{lcccccc}
\toprule
\textbf{Order} & \multicolumn{2}{c}{\textbf{m-sequence}} &  \textbf{Gold code}& \multicolumn{2}{c}{\textbf{m-sequence}} &  \textbf{Gold code} \\
~$k$& $S(i)$  & & $S(i)$ & Tausworthe & & Tausworthe
\\ \midrule
1 & 0.0037  & & -$0.0012$& 0.0003 & &  0.0018\\
2 & 1.0043   & & $1.0011$& 1.0033  & &  0.9992\\
3 & 0.3609  & &  0.0049& 0.0031   & & 0.0022 \\
4 & 3.2182   &&  $3.0061$&  2.8849  &&   2.8421 \\\bottomrule
\end{tabular}
\end{table}


\begin{figure*}[htbp]
\centering
    \includegraphics[scale=0.3]{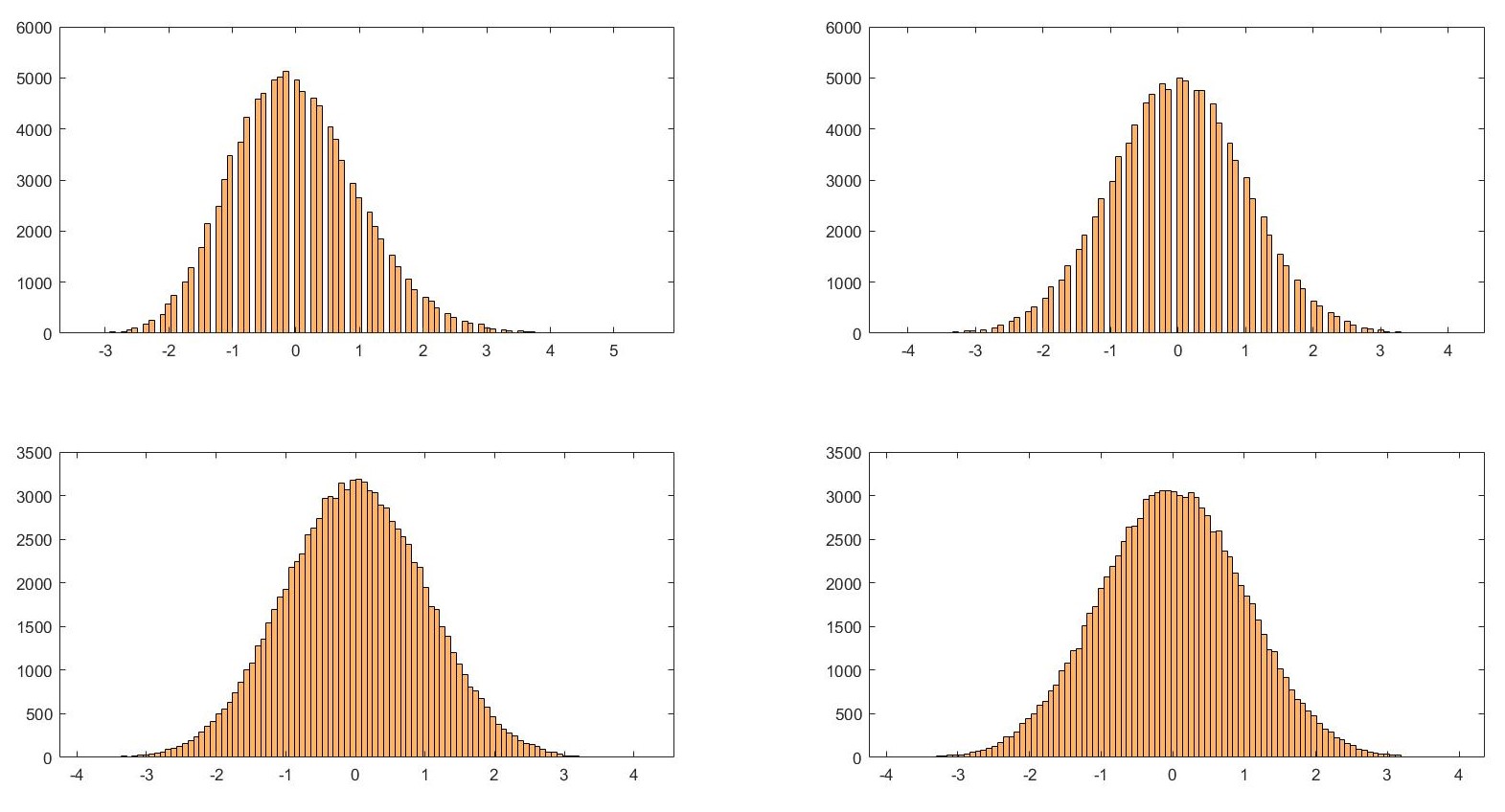}
  \caption{Comparison of histograms with 100 bins generated by GRNGs using the $m$-sequence (left) and the Gold code (right). The histograms for binary sequence model are shown in the top row, while the Tauworthe model is shown in the bottom row.}
  \label{fig: histograms}
\end{figure*}
\begin{figure*}[htbp]
\centering
    \includegraphics[scale=0.3]{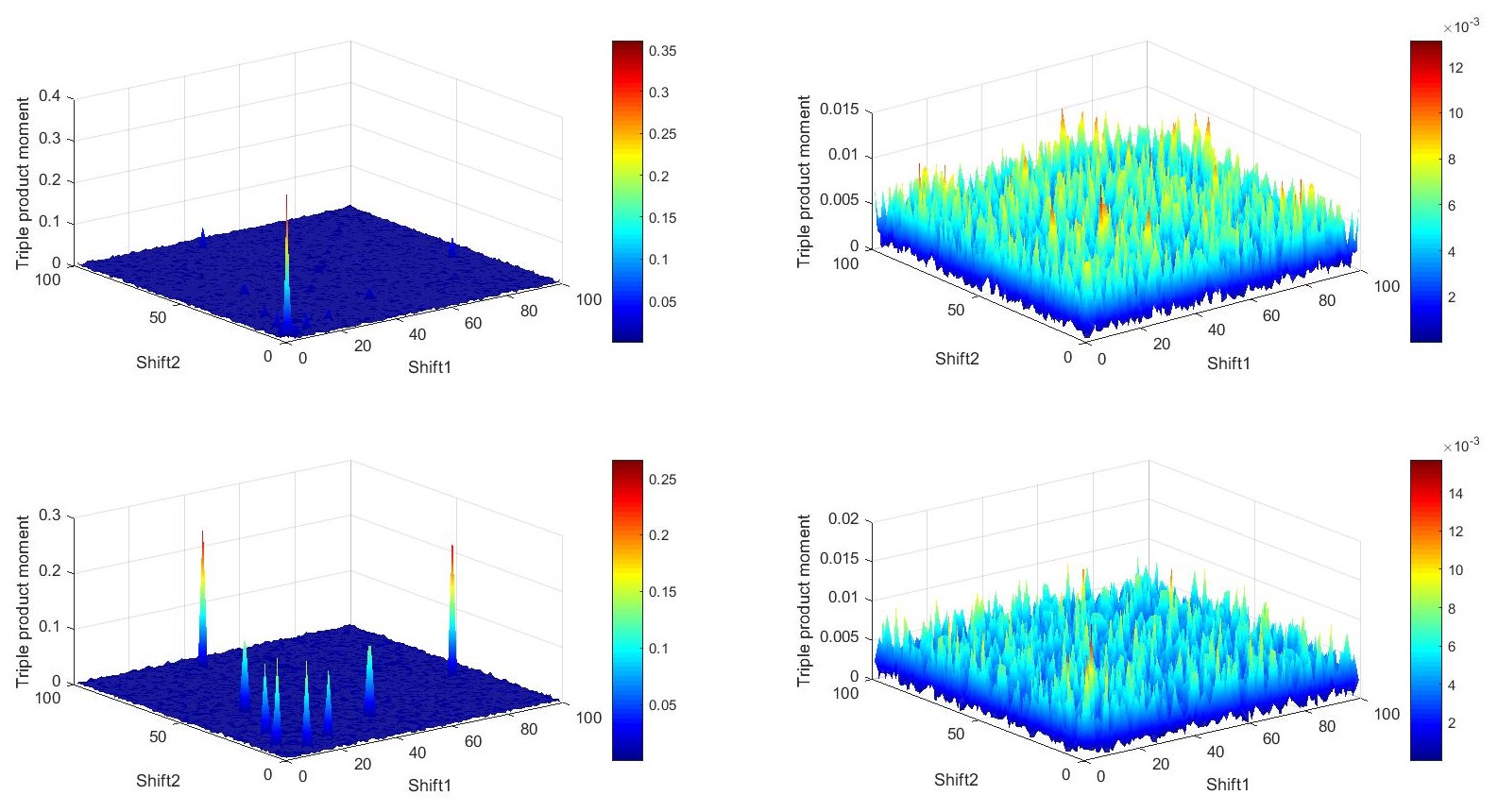}
  \caption{Comparison of the triple product moments (absolute value) generated by GRNGs using the $m$-sequence (left) and the Gold code (right). The calculations for the binary sequence model are shown in the top row, while the Tauworthe model is shown in the bottom row.}
  \label{fig: triple_correlations}
\end{figure*}


Results in Table \ref{tab:moments} shows that the Gold code outperforms the $m$-sequence used as binary sequence. The $m$-sequences show deviations from the expected value in the third and fourth moments. In the case of the GRNGs based on the Tausworthe model there is not a notable difference between using the Gold code and the $m$-sequence.  The fourth moment in the Tausworthe model is far from the expected value $3$ because of the smaller number of sums, which affects the behavior of the tails of the distribution.

 The bit depth and therefore the continuous approximation to the Gaussian distribution is greater for the Tausworthe model as seen in Figure \ref{fig: histograms}. This agrees with previous results of similar GRNG constructions in the literature using $m$-sequences~\cite{kang2010fpga, condo2015pseudo}. Only in the binary sequence method using an $m$-sequence shows a clearly asymmetry due to the presence of peaks on the correlation. 
 

 Figure \ref{fig: triple_correlations} implies patterns in the bispectrum depending on type of binary sequence, independently of the model. The Gold code behaves as expected, close to a Gaussian pseudonoise \cite[Corollary 2]{song2015explicit} while the $m$-sequence show peaks, where  the highest appears in the case of the 
 binary sequence model. 
 
\ Furthermore, we note that the bound of the Theorem \ref{th: th2} normalized by the period $2^{89} - 1$ can be interpreted as the average value of the results obtained for samples of the size $T$ used. For the binary sequence model using the Gold code  it gives us a value close to $10^{-6}$ for $k = 3$ by the Theorem \ref{th: theorem_Gold}  while for the $m$-sequence we get a value much larger than zero \cite[Equation (7d)]{lindholm1968analysis} which partly explains the observed behaviors.
 


\section{Conclusions and Future Work}
\label{sec: conclusions}

We have provided bounds for the higher-order product moments of Gaussian sources constructed from binary sequences.  These bounds depends on the combined correlation measure of order $k$ of the chosen binary sequences. Although these bounds can be improved for special sequences, they are sufficient for real applications and provides the missing link between good binary sequences and Gaussian Random Number Generator. 

In this work we consider Gold codes with a Mersenne prime period, that provides a simple implementation and do not show full-peaks in the third and fourth correlation measures.  We remark that is interesting to characterize other lengths such that the Gold codes presents good properties. However, the length being a Mersenne Prime is not a big restriction for practical purposes. With respect of computational resources, a Gold Code guarantee better statistical properties by doubling the size of the state and the number of XOR gates by four with respect to a $m$-sequence.



Computational results  show that Gold codes offer better statistical and pseudorandom properties regardless of the used model that will depend on the practical application. In the case of the GRNG using an $m$-sequence with the Tausworthe model, we notice that both the moments and the histogram do not show significant deficiencies, but Figure \ref{fig: triple_correlations} shows well-localized peaks, contradicting what is expected for a Gaussian random variable. This negative phenomenon in the case of the $m$-sequence agrees with the existing literature where the distribution of moments is known under the name of distribution of weights of subsequences of $m$-sequences \cite{lindholm1968analysis, jordan1973distribution}. We leave it as an open problem to study how the peaks of triple product moments affect other desirable properties for a GRNG. A first step is to analyze the influence on Gaussian multivariate properties such as orthogonal invariance, i.e., spherical symmetry that may be of special interest in Monte-Carlo and Quasi Monte-Carlo algorithms or sampling in global optimization. 
Also, we propose to study more in detail the Tausworthe model. The standard method to evaluate the Tausworthe model via character sums \cite[Theorem 3.12]{niederreiter1992random} provides no information due to the presence of full peaks in the correlation of order $k=5$ for both sequences.

Finally, it is well known that the peaks of the $m$-sequences depend on the characteristic polynomial that generates them \cite{lindholm1968analysis, jordan1973distribution}, the search of characteristic polynomials that will trade off memory requirements for a improved statistical performance.

\section*{Acknowledgment}
The authors want to thank Andrew Tirkel for pointing the problem and useful discussions. 

Domingo Gómez-Pérez, Ana I. Gómez and Francisco-Javier Soto are partially supported by Research Project
``PROTOCOLOS SEGUROS EN REDES DESCENTRALIZADAS.(AYUDA FINANCIADA CONTRATO PROGRAMA GOB CANTABRIA - UC)''. In addition, Francisco-Javier Soto acknowledges support from the ``PREDOCT2022-006'' of Univerdidad Rey Juan Carlos.



\bibliographystyle{IEEEtran}
\bibliography{bibliografia}

\end{document}